\documentclass[12pt,reqno]{amsart}

\usepackage[T1]{fontenc}
\usepackage{lmodern}
\usepackage{microtype}
\usepackage{geometry}
\usepackage{amsmath,amssymb,amsthm,mathtools}
\usepackage{enumitem}
\usepackage{xcolor}
\usepackage[colorlinks=true,linkcolor=blue!45!black,citecolor=blue!45!black,urlcolor=blue!45!black]{hyperref}

\allowdisplaybreaks
\setlist[itemize]{topsep=3pt,itemsep=2pt,parsep=0pt}
\setlist[enumerate]{topsep=3pt,itemsep=2pt,parsep=0pt}

\newtheorem{theorem}{Theorem}[section]
\newtheorem{proposition}[theorem]{Proposition}
\newtheorem{corollary}[theorem]{Corollary}
\newtheorem{lemma}[theorem]{Lemma}
\theoremstyle{definition}
\newtheorem{definition}[theorem]{Definition}
\newtheorem{example}[theorem]{Example}
\theoremstyle{remark}
\newtheorem{remark}[theorem]{Remark}

\newcommand{\kk}{\Bbbk}
\newcommand{\Der}{\operatorname{Der}}
\newcommand{\Sym}{\operatorname{Sym}}
\newcommand{\id}{\operatorname{id}}
\newcommand{\tr}{\operatorname{tr}}
\newcommand{\Pcal}{\mathcal P}
\newcommand{\Fcot}{F_{\mathrm{cot}}}

\title[From Lie--Rinehart algebras to \texorpdfstring{$F$}{F}-manifold algebras]{From Lie--Rinehart Algebras to \texorpdfstring{$F$}{F}-Manifold Algebras}

\author{Yufeng Pei}
\address{Department of Mathematics, Huzhou Normal University, Huzhou 313000, Zhejiang, China}
\email{pei@huznu.edu.cn}

\author{Yunhe Sheng}
\address{Department of Mathematics, Jilin University, Changchun 130012, Jilin, China}
\email{shengyh@jlu.edu.cn}

\begin{document}

	\begin{abstract}
		For every Lie--Rinehart algebra, we construct an $F$-manifold algebra on the
		direct sum of its base algebra and module.   Contrary to the assertion in
		\cite[Proposition 13.3.26]{LodayVallette}, the resulting structure is generally
		not Poisson.  We determine when powers of the positive-degree ideal in the
		associated symmetric Poisson algebra are Poisson ideals, and relate the
		Leibnizator to the Lie--Rinehart differential.  For a finite projective module
		of constant rank, the trace of the Leibnizator recovers the anchor and yields a
		rigidity result for injective anchors.  We conclude with algebraic and
		geometric examples.
	\end{abstract}
	

	\maketitle

	\section{Introduction}
	
	Lie--Rinehart algebras were introduced by Rinehart as $(R,A)$-Lie algebras
	\cite{Rinehart}.  The present terminology was later used by Huebschmann
	\cite{Huebschmann1990,Huebschmann1998}.  These
	algebras are algebraic counterparts of Lie algebroids \cite{Mackenzie} and
	occur in Poisson geometry, symplectic geometry, and quantization
	\cite{Kosmann1996}.  A Lie--Rinehart algebra
	consists of a commutative algebra $A$, an $A$-module $L$, a Lie bracket on
	$L$, and an anchor $\rho:L\to\Der_\kk(A)$ satisfying the Leibniz rule.  The
	bracket extends to a Poisson bracket on the symmetric algebra $\Sym_A(L)$.
	For the Lie--Rinehart algebra associated with a Lie algebroid $E\to M$, this
	algebra is naturally identified with the fiberwise polynomial functions on
	$E^*$, equipped with the bracket induced by its linear Poisson structure.
	
	Dubrovin introduced Frobenius manifolds as a geometric formulation of the
	Witten--Dijkgraaf--Verlinde--Verlinde associativity equations in
	two-dimensional topological field theory \cite{Dub95}.  Hertling and Manin
	then introduced $F$-manifolds by weakening the conditions of a Frobenius
	manifold \cite{HertlingManin}.  Every Frobenius manifold is an
	$F$-manifold.  Related algebraic and geometric aspects of $F$-manifolds
	have been studied in \cite{BCMM,CT,CLD,DS11,Laubie,LPR11,Merku,TV}.
	Dotsenko introduced $F$-manifold algebras and identified their operad with the
	associated graded operad of the pre-Lie operad under the filtration by powers
	of the Lie bracket \cite{Dot}.  It was shown in \cite{LiuShengBai2020} that
	$F$-manifold algebras arise as semiclassical limits of pre-Lie deformations
	of commutative associative algebras.
	
	Our main result associates an $F$-manifold algebra with every Lie--Rinehart
	algebra $(A,L)$.  On $A\oplus L$, we use the product
	$(a+X)\bullet(b+Y)=ab+aY+bX$ and the semidirect Lie bracket $[-,-]$ determined by the
	anchor and the bracket on $L$.  The Leibnizator
$$
\Pcal_{a+X}(b+Y,c+Z):=[a+X,(b+Y)\bullet(c+Z)]-(b+Y)\bullet[a+X,c+Z]-(c+Z)\bullet[a+X,b+Y]
$$
is given by
	\[
	\Pcal_{a+X}(b+Y,c+Z)=\rho(Y)(a)Z+\rho(Z)(a)Y.
	\]
	This formula reduces the Hertling--Manin identity to the derivation property
	of the anchor.  In general, the Leibnizator does not vanish, so the resulting
	algebra need not to be Poisson.
	
	As commutative algebras, $A\oplus L$ is isomorphic to
	$\Sym_A(L)/I^2$, where $I$ is the positive-degree ideal of the commutative associative
	algebra $\Sym_A(L)$.  The ideal $I^2$ need not to be a Poisson ideal.  We extend
	this observation to every power $I^n$.  If $L$ is a finite projective
	$A$-module of positive constant rank, then $I^n$ is a Poisson ideal for some
	$n\geq1$ (equivalently, for every $n\geq1$) if and only if the anchor is zero.
	We also express the
	Leibnizator in terms of the Lie--Rinehart differential.  The trace of the
	Leibnizator recovers the anchor and gives a rigidity statement when the
	anchor is injective.
	
%

	The paper is organized as follows.  In Section \ref{sec:construction}, we
	construct the $F$-manifold algebra on $A\oplus L$, compute its Leibnizator,
	and prove the restricted recognition result.  In Section
	\ref{sec:ideal}, we study the quotients of $\Sym_A(L)$ by powers of the
	positive-degree ideal and recover the anchor from the Leibnizator.  In
	Section \ref{sec:examples}, we give algebraic and geometric examples.
	
	Throughout, $\kk$ is a field of characteristic zero.  All commutative algebras
	are assumed to be unital, and all tensor products are over $\kk$ unless
	otherwise stated.
	
	\section{The associated \texorpdfstring{$F$}{F}-manifold algebra}
	\label{sec:construction}
	
	In this section, we show that $A\oplus L$ is an \texorpdfstring{$F$}{F}-manifold algebra associated to any Lie--Rinehart algebra $(A,L)$.
	
	\begin{definition}
		Let $A$ be a commutative associative $\kk$-algebra.  A \textbf{Lie--Rinehart
			algebra over $A$} consists of an $A$-module $L$, a $\kk$-Lie bracket
		$[-,-]_L$ on $L$, and an $A$-linear Lie algebra homomorphism
		$
		\rho:L\longrightarrow \Der_{\kk}(A),
		$
		called the anchor, such that
		\begin{equation}\label{eq:LR}
			[X,aY]_L=a[X,Y]_L+\rho(X)(a)Y
		\end{equation}
		for all $a\in A$ and $X,Y\in L$.
	\end{definition}
	
	We denote a Lie--Rinehart algebra by $(A,L,[-,-]_L,\rho)$, or simply by
	$(A,L)$.

	\begin{example}[Ordinary Lie algebras]
		If $A=\kk$, then $\Der_\kk(A)=0$ and a Lie--Rinehart algebra over $A$ is
		just an ordinary Lie algebra $L$ over $\kk$.  The $A$-module structure is the
		usual scalar multiplication and the anchor is zero.
	\end{example}
	
	\begin{example}[Derivation Lie--Rinehart algebra]
		For every commutative algebra $A$, the pair
		\[
		(A,\Der_\kk(A))
		\]
		is a Lie--Rinehart algebra.  $\Der_\kk(A)$ has the natural $A$-module structure, the
		anchor is the identity map, and the Lie bracket is the usual commutator
		bracket.
	\end{example}
	
	\begin{example}[Lie $A$-algebras]
		If the anchor is zero, then \eqref{eq:LR} says
		\[
		[X,aY]_L=a[X,Y]_L.
		\]
		Thus a zero-anchor Lie--Rinehart algebra is a Lie algebra object in
		$A$-modules.
	\end{example}
	
	A Poisson algebra is a commutative associative algebra $A$ equipped with a
	Lie bracket $\{-,-\}$ satisfying
	\[
	\{f,gh\}=\{f,g\}h+g\{f,h\},
	\qquad f,g,h\in A.
	\]
	
	\begin{example}[K\"ahler differentials of a Poisson algebra]
		Let $(A,\{-,-\})$ be a Poisson algebra and let $\Omega^1_{A/\kk}$ be the
		module of K\"ahler differentials.  Then $(A,\Omega^1_{A/\kk})$ is a
		Lie--Rinehart algebra whose anchor $\pi^\sharp$ is determined by
		\[
		\pi^\sharp(df)(g)=\{f,g\},
		\]
		and the Lie bracket is determined by
		\[
		[df,dg]_\pi=d\{f,g\}.
		\]
	\end{example}

	\begin{definition}
		An \textbf{$F$-manifold algebra} \cite{Dot,HertlingManin,LiuShengBai2020} is a
		triple $(B,\bullet,[-,-])$ such that
		$(B,\bullet)$ is a commutative associative algebra, $(B,[-,-])$ is a Lie
		algebra, and the Leibnizator
		\begin{equation}\label{eq:Leibnizator-def}
			\Pcal_x(y,z):=[x,y\bullet z]-[x,y]\bullet z-y\bullet[x,z]
		\end{equation}
		satisfies the Hertling--Manin identity
		\begin{equation}\label{eq:HM}
			\Pcal_{x\bullet y}(z,w)=x\bullet\Pcal_y(z,w)+y\bullet\Pcal_x(z,w)
		\end{equation}
		for all $x,y,z,w\in B$.
	\end{definition}

	\begin{example}[Poisson algebras]
		Every Poisson algebra is an $F$-manifold algebra.  In this case, the
		Leibnizator vanishes identically.
	\end{example}

	Let $(A,L,[-,-]_L,\rho)$ be a Lie--Rinehart algebra and set $P=A\oplus L$.
	Define a product on $P$ by
	\begin{equation}\label{eq:F-product}
		(a+X)\bullet(b+Y)=ab+aY+bX,
		\qquad a,b\in A,\quad X,Y\in L,
	\end{equation}
	Thus $A$ is a subalgebra of $P$, $L$ is an ideal with
	$L\bullet L=0$, and $a\bullet X=aX$ for $a\in A$ and $X\in L$.
	Define a bracket on $P$ by
	\begin{equation}\label{eq:F-bracket}
		[a+X,b+Y]
		=\rho(X)(b)-\rho(Y)(a)+[X,Y]_L,
		\qquad a,b\in A,\quad X,Y\in L.
	\end{equation}
	Equivalently,
	\[
	[A,A]=0,
	\qquad
	[X,a]=\rho(X)(a),
	\qquad
	[X,Y]=[X,Y]_L.
	\]
	This is the semidirect bracket determined by the anchor action on \(A\) and
	the Lie bracket on \(L\).

	We compute the Leibnizator of the bracket with respect to $\bullet$.
	
	\begin{lemma}\label{lem:Leibnizator-formula}
		For $a+X, b+Y, c+Z\in A\oplus L$, we have
		\begin{equation}\label{eq:Leibnizator-formula}
			\Pcal_{a+X}(b+Y,c+Z)=\rho(Y)(a)Z+\rho(Z)(a)Y.
		\end{equation}
		
	\end{lemma}
	
	\begin{proof}
		We expand the three terms in the definition of $\Pcal$.  First,
		\[
		(b+Y)\bullet(c+Z)=bc+bZ+cY.
		\]
		Using \eqref{eq:F-bracket}, the derivation property of $\rho(X)$, the
		$A$-linearity of the anchor, and the Lie--Rinehart identity, we get
		\begin{align*}
			[a+X,(b+Y)\bullet(c+Z)]
			&= [a+X,bc+bZ+cY] \\
			&=\rho(X)(bc)-\rho(bZ+cY)(a)+[X,bZ+cY]_L \\
			&=\rho(X)(b)c+b\rho(X)(c)-b\rho(Z)(a)-c\rho(Y)(a) \\
			&\quad +b[X,Z]_L+c[X,Y]_L+\rho(X)(b)Z+\rho(X)(c)Y.
		\end{align*}
		Second,
		\begin{align*}
			[a+X,b+Y]\bullet(c+Z)
			&=\bigl(\rho(X)(b)-\rho(Y)(a)+[X,Y]_L\bigr)\bullet(c+Z) \\
			&=c\rho(X)(b)-c\rho(Y)(a)+\rho(X)(b)Z-\rho(Y)(a)Z+c[X,Y]_L.
		\end{align*}
		Third,
		\begin{align*}
			(b+Y)\bullet[a+X,c+Z]
			&=(b+Y)\bullet\bigl(\rho(X)(c)-\rho(Z)(a)+[X,Z]_L\bigr) \\
			&=b\rho(X)(c)-b\rho(Z)(a)+\rho(X)(c)Y-\rho(Z)(a)Y+b[X,Z]_L.
		\end{align*}
		Subtracting the last two displayed expressions from the first cancels all
		terms except
		\[
		\rho(Y)(a)Z+\rho(Z)(a)Y.
		\]
		This proves \eqref{eq:Leibnizator-formula}.
	\end{proof}
	
	\begin{theorem}\label{thm:LR-to-FMan}
		Let $(A,L)$ be a Lie--Rinehart algebra.  Then
		$(A\oplus L,\bullet,[-,-])$ is an $F$-manifold algebra.
	\end{theorem}
	
	\begin{proof}
		By \eqref{eq:F-product}, the product $\bullet$ is commutative and associative.
		The bracket is the semidirect product bracket
		for the action $\rho:L\to\Der_\kk(A)$.  It is therefore a Lie bracket: the
		Jacobi identity on $L$ is inherited from $[-,-]_L$, the identity involving two
		elements of $L$ and one element of $A$ follows because $\rho$ is a Lie algebra
		homomorphism, and all remaining components vanish.  It remains to prove the
		Hertling--Manin identity
		\begin{equation}\label{eq:HM-proof-target}
			\Pcal_{(a+X)\bullet (b+Y)}(c+Z,d+W)
			=(a+X)\bullet\Pcal_{b+Y}(c+Z,d+W)+(b+Y)\bullet\Pcal_{a+X}(c+Z,d+W).
		\end{equation}
		
		By Lemma \ref{lem:Leibnizator-formula}, we have
		\begin{align*}
			\Pcal_{(a+X)\bullet (b+Y)}(c+Z,d+W)
			&=\rho(Z)(ab)W+\rho(W)(ab)Z \\
			&=\bigl(a\rho(Z)(b)+b\rho(Z)(a)\bigr)W
			+\bigl(a\rho(W)(b)+b\rho(W)(a)\bigr)Z,
		\end{align*}
		where we used that $\rho(Z)$ and $\rho(W)$ are derivations of $A$.
		
		On the other hand,
		\[
		\Pcal_{b+Y}(c+Z,d+W)=\rho(Z)(b)W+\rho(W)(b)Z\in L,
		\]
		and similarly
		\[
		\Pcal_{a+X}(c+Z,d+W)=\rho(Z)(a)W+\rho(W)(a)Z\in L.
		\]
		Since an element $a+X\in P$ acts by the product $\bullet$ on $L$ only through
		its $A$-component, we have
		\begin{align*}
			(a+X)\bullet\Pcal_{b+Y}(c+Z,d+W)
			&=a\rho(Z)(b)W+a\rho(W)(b)Z,\\
			(b+Y)\bullet\Pcal_{a+X}(c+Z,d+W)
			&=b\rho(Z)(a)W+b\rho(W)(a)Z.
		\end{align*}
		Adding the two equalities gives the expression for
		$\Pcal_{(a+X)\bullet (b+Y)}(c+Z,d+W)$.  Hence \eqref{eq:HM-proof-target} holds.
	\end{proof}

	The following corollary characterizes the Poisson case.
	
	\begin{corollary}\label{cor:poisson-specialization}
		The associated $F$-manifold algebra $(A\oplus L,\bullet,[-,-])$ is a
		Poisson algebra if and only if
		\begin{equation}\label{eq:Poisson-condition}
			\rho(Y)(a)Z+\rho(Z)(a)Y=0
		\end{equation}
		for all $a\in A$ and $Y,Z\in L$.
	\end{corollary}
	
	\begin{proof}
		For a commutative associative algebra with a Lie bracket, the Poisson identity
		is equivalent to the vanishing of the Leibnizator.  The assertion therefore
		follows directly from Lemma \ref{lem:Leibnizator-formula}.
	\end{proof}
	
	\begin{remark}\label{rem:LV-poisson-claim}
		It is asserted in \cite[Proposition 13.3.26]{LodayVallette} that $A\oplus L$,
		with the product \eqref{eq:F-product} and bracket \eqref{eq:F-bracket}, is a
		Poisson algebra for an arbitrary Lie--Rinehart algebra.  Corollary
		\ref{cor:poisson-specialization} shows that the additional condition
		\eqref{eq:Poisson-condition} is required.  Theorem
		\ref{thm:LR-to-FMan} gives the valid conclusion without this condition.
	\end{remark}

	\section{Poisson truncations and the Lie--Rinehart differential}
	\label{sec:ideal}
	
	We now compare this construction with truncations of the canonical
	Poisson algebra associated with a Lie--Rinehart algebra.  This comparison
	also identifies the higher-order obstruction to obtaining a Poisson
	quotient.
	
	Let $(A,L)$ be a Lie--Rinehart algebra and write
	\[
	\Sym_A(L)=\bigoplus_{n\geq0}\Sym_A^n(L).
	\]
	We use the following standard linear Poisson construction; see, for example,
	\cite{Huebschmann1990,Huebschmann1998}.
	
	\begin{proposition}\label{prop:linear-poisson-prelim}
		The symmetric algebra $\Sym_A(L)$ carries a unique Poisson bracket
		determined on generators by
		\begin{equation}\label{eq:linear-poisson}
			\{a,b\}=0,
			\qquad
			\{X,a\}=\rho(X)(a),
			\qquad
			\{X,Y\}=[X,Y]_L,
		\end{equation}
		for $a,b\in A$ and $X,Y\in L$.  Moreover,
		\[
		\{\Sym_A^p(L),\Sym_A^q(L)\}\subseteq \Sym_A^{p+q-1}(L).
		\]
		Here $\Sym_A^{-1}(L)=0$.
	\end{proposition}
	
	\subsection{First-order and higher truncations}
	
	Let
	\[
	I=\bigoplus_{n\geq1}\Sym_A^n(L)
	\]
	be the positive-degree ideal of $\Sym_A(L)$.
	
	\begin{theorem}\label{thm:first-order-truncation}
		As commutative algebras,
		\[
		\Sym_A(L)/I^2\simeq A\oplus L,
		\]
		with the product \eqref{eq:F-product}.  The ideal $I^2$ is a Poisson ideal of
		$\Sym_A(L)$ if and only if
		\[
		\Pcal_a(X,Y)=\rho(Y)(a)X+\rho(X)(a)Y=0
		\]
		for all $a\in A$ and $X,Y\in L$.
	\end{theorem}
	
	\begin{proof}
		The quotient statement follows from the symmetric grading:
		\[
		\Sym_A(L)/I^2\simeq \Sym_A^0(L)\oplus \Sym_A^1(L)=A\oplus L.
		\]
		The induced multiplication is \eqref{eq:F-product}.
		
		For $X,Y\in L$ and $a\in A$, the Poisson Leibniz rule in $\Sym_A(L)$ gives
		\[
		\{XY,a\}=X\{Y,a\}+Y\{X,a\}
		=\rho(Y)(a)X+\rho(X)(a)Y.
		\]
		By Lemma \ref{lem:Leibnizator-formula}, the same expression is
		$\Pcal_a(X,Y)$ in the $F$-manifold algebra $A\oplus L$.  It has
		degree one.  Hence a nonzero value gives
		$\{I^2,\Sym_A(L)\}\not\subseteq I^2$.
		
		Conversely, assume that the displayed expression vanishes for all $a,X,Y$.
		Since the linear Poisson bracket has  degree $-1$, it is enough, by
		the biderivation property, to check brackets of degree-two monomials with
		algebra generators of degrees $0$ and $1$.  The degree-zero case is the
		assumption.  For $Z\in L$,
		\[
		\{XY,Z\}=X\{Y,Z\}+Y\{X,Z\}=X[Y,Z]_L+Y[X,Z]_L\in I^2.
		\]
		Since $I^2$ is generated by degree-two monomials, this proves
		$\{I^2,\Sym_A(L)\}\subseteq I^2$.
	\end{proof}

	\begin{remark}
		This gives another interpretation of Remark \ref{rem:LV-poisson-claim}: the
		linear Poisson bracket need not descend to the first-order truncation.
	\end{remark}
	
	The Lie--Rinehart differential on functions is the $\kk$-linear map
	\[
	d_L:A\longrightarrow L^\vee:=\operatorname{Hom}_A(L,A),
	\qquad
	(d_La)(X)=\rho(X)(a).
	\]
	For $n\geq1$, define an $A$-linear map
	\[
	j_n:L^\vee\longrightarrow
	\operatorname{Hom}_A\bigl(\Sym_A^n(L),\Sym_A^{n-1}(L)\bigr)
	\]
	by
	\begin{equation}\label{eq:jn-map}
		j_n(\alpha)(X_1\cdots X_n)
		=\sum_{i=1}^n \alpha(X_i)
		X_1\cdots\widehat{X_i}\cdots X_n,
	\end{equation}
	where the empty product for $n=1$ is $1\in A$.  The right-hand side is
	symmetric and $A$-multilinear in $X_1,\ldots,X_n$, so $j_n$ is well defined.
	
	\begin{proposition}\label{prop:higher-truncations}
		For every $n\geq1$ and $a\in A$,
		\begin{equation}\label{eq:higher-bracket-differential}
			\{X_1\cdots X_n,a\}
			=j_n(d_La)(X_1\cdots X_n).
		\end{equation}
		Consequently, $I^n$ is a Poisson ideal of $\Sym_A(L)$ if and only if
		\begin{equation}\label{eq:higher-poisson-condition}
			j_n(d_La)=0
			\qquad\text{for every }a\in A.
		\end{equation}
		If $L$ is a finite projective $A$-module of positive constant rank, then, for
		each $n\geq1$, condition \eqref{eq:higher-poisson-condition} is equivalent to
		$\rho=0$.
	\end{proposition}
	
	\begin{proof}
		Formula \eqref{eq:higher-bracket-differential} follows by applying the Poisson
		Leibniz rule to the product $X_1\cdots X_n$.  Since
		\[
		I^n=\bigoplus_{m\geq n}\Sym_A^m(L)
		\]
		and the bracket has  degree $-1$, only the bracket of its degree-$n$
		part with $A$ can leave $I^n$.  This bracket vanishes precisely when
		\eqref{eq:higher-poisson-condition} holds; the biderivation property and the
		fact that $A\cup L$ generates $\Sym_A(L)$ then give the claimed criterion.
		
		Now suppose that $L$ is a finite projective $A$-module of constant rank
		$r>0$.  The map $j_n$ is injective.  Indeed, after localization at any prime ideal
		$\mathfrak p$ of $A$, choose a basis $e_1,\ldots,e_r$ of $L_{\mathfrak p}$.
		If $j_n(\alpha)=0$, then
		\[
		0=j_n(\alpha_{\mathfrak p})(e_i^n)
		=n\alpha_{\mathfrak p}(e_i)e_i^{n-1},
		\]
		so $\alpha_{\mathfrak p}(e_i)=0$ for every $i$; here $n$ is invertible
		because $\operatorname{char}\kk=0$ and $e_i^{n-1}$ is a basis monomial.
		Thus $\alpha$ vanishes after localization at every prime ideal and hence is zero.
		It follows that \eqref{eq:higher-poisson-condition} is
		equivalent to $d_La=0$ for every $a$, which is equivalent to $\rho=0$.
	\end{proof}
	
	\begin{remark}\label{rem:coisotropic-not-poisson}
		The positive-degree ideal always satisfies $\{I,I\}\subseteq I$. This is weaker than being a
		Poisson ideal.  If $L$ is a finite projective $A$-module of positive constant
		rank and the anchor is nonzero, Proposition \ref{prop:higher-truncations}
		shows that none of the quotients $\Sym_A(L)/I^n$ inherits the linear Poisson
		bracket.
	\end{remark}
	
	\begin{remark}\label{rem:graded-poisson}
		More generally, let $B=\bigoplus_{m\geq0}B_m$ be a graded commutative Poisson
		algebra whose bracket has degree $-1$.  Then $A=B_0$ and $L=B_1$ form a
		Lie--Rinehart algebra: the product gives the $A$-module structure, the
		restrictions of the bracket give the anchor and the bracket on $L$, and the
		Poisson Leibniz rule gives the Lie--Rinehart identities.  Equip $B_0\oplus B_1$
		with the product $(a+X)\bullet(b+Y)=ab+aY+bX$, where $a,b\in B_0$ and
		$X,Y\in B_1$.  This gives the $F$-manifold algebra associated with the
		Lie--Rinehart structure above.
	\end{remark}
	
	\subsection{The Leibnizator and recovery of the anchor}
	
	For $n=2$, write $j=j_2$; explicitly, this is the $A$-linear map
	\[
	j:L^\vee\longrightarrow \operatorname{Hom}_A(\Sym_A^2(L),L)
	\]
	defined by
	\begin{equation}\label{eq:j-map}
		j(\alpha)(X,Y)=\alpha(X)Y+\alpha(Y)X.
	\end{equation}
	
	\begin{proposition}
		\label{prop:leibnizator-differential}
		Let $(A\oplus L,\bullet,[-,-])$ be the $F$-manifold algebra associated with a
		Lie--Rinehart algebra $(A,L)$.  Then the Leibnizator has only one possible
		nonzero component,
		\[
		A\times\Sym_A^2(L)\longrightarrow L,
		\]
		and it is given by
		\begin{equation}\label{eq:P-jd}
			\Pcal_a=j(d_La),\qquad a\in A.
		\end{equation}
		
		Consequently, the Hertling--Manin identity follows from the Leibniz rule
		\[
		d_L(ab)=a\,d_Lb+b\,d_La.
		\]
	\end{proposition}
	
	\begin{proof}
		
		Formula \eqref{eq:P-jd} is the same identity written using \eqref{eq:j-map}.
		For the last assertion, compute
		\[
		\Pcal_{ab}=j(d_L(ab))=j(a\,d_Lb+b\,d_La)=a\Pcal_b+b\Pcal_a.
		\]
		This identity gives the Hertling--Manin identity for
		$(A\oplus L,\bullet,[-,-])$.
	\end{proof}
	
	\medskip\noindent\textbf{Trace formula.}
	Assume that $L$ is a finite projective $A$-module of constant rank $r$.  Since
	$\operatorname{char}\kk=0$, the scalar $r+1$ is invertible.  For
	$T\in\operatorname{Hom}_A(\Sym_A^2(L),L)$, define
	$\tr:\operatorname{Hom}_A(\Sym_A^2(L),L)\to L^\vee$ by
	\[
	(\tr T)(X)=\operatorname{Tr}\bigl(Y\mapsto T(X,Y)\bigr),
	\]
	where $\operatorname{Tr}:\operatorname{Hom}_A(L,L)\to A$ is the usual trace.
	
	\begin{proposition}\label{prop:trace-formula}
		For every $\alpha\in L^\vee$,
		\begin{equation}\label{eq:trace-j}
			\tr(j(\alpha))=(r+1)\alpha.
		\end{equation}
		Consequently,
		\begin{equation}\label{eq:trace-P}
			\tr(\Pcal_a)=(r+1)d_La,
		\end{equation}
		and hence
		\begin{equation}\label{eq:anchor-trace-formula}
			\rho(X)(a)=\frac1{r+1}\operatorname{Tr}\bigl(Y\mapsto \Pcal_a(X,Y)\bigr).
		\end{equation}
	\end{proposition}
	
	\begin{proof}
		The endomorphism $Y\mapsto j(\alpha)(X,Y)$ is
		\[
		Y\mapsto \alpha(X)Y+\alpha(Y)X.
		\]
		The first summand is scalar multiplication by $\alpha(X)$ and has trace
		$r\alpha(X)$, since $\operatorname{Tr}(\id_L)=r$.  The second summand is
		the rank-one endomorphism $Y\mapsto \alpha(Y)X$ and has trace
		$\alpha(X)$.  In fact, since $L$ is finitely generated projective, the dual
		basis lemma gives elements $e_i\in L$ and $e^i\in L^\vee$ such that
		$Y=\sum_i e^i(Y)e_i$ for all $Y\in L$.  The trace of
		$\phi\in\operatorname{End}_A(L)$ is
		$\operatorname{Tr}(\phi)=\sum_i e^i(\phi(e_i))$, independently of the
		chosen dual basis.  For $\phi(Y)=\alpha(Y)X$, this gives
		\[
		\operatorname{Tr}(\phi)
		=\sum_i e^i\bigl(\alpha(e_i)X\bigr)
		=\sum_i \alpha(e_i)e^i(X)
		=\alpha\Bigl(\sum_i e^i(X)e_i\Bigr)
		=\alpha(X).
		\]
		Equivalently, $\phi$ is the image of $X\otimes\alpha$ under
		$L\otimes_A L^\vee\simeq\operatorname{End}_A(L)$, and the trace is the
		evaluation pairing $X\otimes\alpha\mapsto\alpha(X)$.  This proves
		\eqref{eq:trace-j}.  The remaining assertions follow from
		Proposition \ref{prop:leibnizator-differential}.
	\end{proof}
	
	\begin{corollary}\label{cor:faithful-rigidity}
		Let $L$ be a finite projective $A$-module of constant rank.  Consider two
		Lie--Rinehart structures on the same pair $(A,L)$, and denote the
		Leibnizators of their associated $F$-manifold algebras by $\Pcal^{(1)}$ and
		$\Pcal^{(2)}$.  If $\Pcal^{(1)}=\Pcal^{(2)}$, then the two anchors coincide.
		If this common anchor is injective, then the two Lie--Rinehart structures
		coincide.
	\end{corollary}
	
	\begin{proof}
		Formula \eqref{eq:anchor-trace-formula} recovers the anchor from the
		Leibnizator, so the anchors coincide.  Let $\rho$ be their common anchor and
		write the two brackets as $[-,-]_1$ and $[-,-]_2$.  Since the anchor of a
		Lie--Rinehart algebra is a Lie algebra homomorphism,
		\[
		\rho([X,Y]_1)=[\rho(X),\rho(Y)]=\rho([X,Y]_2).
		\]
		Injectivity of $\rho$ therefore implies $[X,Y]_1=[X,Y]_2$.
	\end{proof}

	We finish the section with two canonical Poisson subalgebras of the associated
	$F$-manifold algebra.  Here a Poisson subalgebra means a subspace closed under
	both operations, with vanishing restricted Leibnizator.  Let
	\[
	A^L=\{a\in A\mid \rho(X)(a)=0\text{ for all }X\in L\}
	\]
	be the invariant subalgebra, and let
	\[
	K=\ker\rho\subseteq L.
	\]
	
	\begin{proposition}\label{prop:Poisson-subalgebras}
		The subspaces
		\[
		A^L\oplus L
		\qquad\text{and}\qquad
		A\oplus K
		\]
		are Poisson subalgebras of the $F$-manifold algebra
		$(A\oplus L,\bullet,[-,-])$.  Their intersection is $A^L\oplus K$.
	\end{proposition}
	
	\begin{proof}
		For \(A^L\oplus L\), the product is closed because \(A^L\) is a subalgebra of
		\(A\).  It is closed under the bracket because the \(A\)-component of
		\([a+X,b+Y]\) is
		$\rho(X)(b)-\rho(Y)(a)$, which vanishes for $a,b\in A^L$.  Its Leibnizator
		is zero by Proposition \ref{prop:leibnizator-differential}, since $d_La=0$ for
		$a\in A^L$.
		
		For \(A\oplus K\), the product is closed because \(K\) is an \(A\)-submodule
		of \(L\).  Since \(\rho\) is an \(A\)-linear Lie algebra homomorphism, \(K\)
		is also closed under the Lie bracket.  Hence \(A\oplus K\) is closed under the
		bracket.  Its Leibnizator is zero because every \(L\)-component appearing in
		the second and third arguments lies in \(K\) and hence has zero anchor.  The
		intersection is \(A^L\oplus K\) by definition.
	\end{proof}
	
	\begin{remark}
		Geometrically, $A^L$ consists of functions constant along the anchor
		distribution, while $K=\ker\rho$ is the isotropy of the anchor.
	\end{remark}

	\section{Examples}\label{sec:examples}
	
	We illustrate the construction with algebraic and geometric examples.

	\subsection{Lie algebroids}
	
	Let $E\to M$ be a Lie algebroid over a smooth manifold $M$, with anchor
	$\rho:E\to TM$ and bracket $[-,-]_E$ on sections.  Then
	\[
	A=C^\infty(M),
	\qquad
	L=\Gamma(E)
	\]
	is a Lie--Rinehart algebra.  The construction above equips the vector space
	\[
	C^\infty(M)\oplus\Gamma(E)
	\]
	with the $F$-manifold algebra structure whose product is
	\begin{equation}\label{eq:lie-algebroid-F-product}
		(f+X)\bullet(g+Y)=fg+fY+gX
	\end{equation}
	and whose bracket is
	\begin{equation}\label{eq:lie-algebroid-F-bracket}
		[f+X,g+Y]
		=\rho(X)(g)-\rho(Y)(f)+[X,Y]_E.
	\end{equation}
	Its Leibnizator is
	\begin{equation}\label{eq:lie-algebroid-Leibnizator}
		\Pcal_{f+X}(g+Y,h+Z)=\rho(Y)(f)Z+\rho(Z)(f)Y.
	\end{equation}
	The linear Poisson structure associated with a Lie algebroid is defined on
	\(E^*\), and \(C^\infty(M)\oplus\Gamma(E)\) is its truncation to fiberwise
	polynomial functions of degree at most one.  The tangent Lie algebroid
	\(E=TM\) gives
	\[
	C^\infty(M)\oplus\mathfrak X(M),
	\]
	while the cotangent Lie algebroid of a Poisson manifold gives the example
	studied below.

	\subsection{The polynomial derivation algebra}
	
	Let
	$
	A=\kk[x]$ and $
	L=A D$, where $
	D=\frac{d}{dx}.
	$
The anchor is the identity map $A D\to\Der_\kk(A)$.  The Lie--Rinehart bracket is
	\[
	[fD,gD]_L=(fg'-gf')D.
	\]
	The associated $F$-manifold algebra is
$
	P=A\oplus A D.
$
	For $a,b,f,g\in A$,
	\begin{equation}\label{eq:poly-product}
		(a+fD)\bullet(b+gD)=ab+(ag+bf)D,
	\end{equation}
	and
	\begin{equation}\label{eq:poly-bracket}
		[a+fD,b+gD]=fb'-ga'+(fg'-gf')D.
	\end{equation}
	This is the semidirect product bracket associated with the identity anchor.
	
	The Leibnizator is
	\begin{equation}\label{eq:poly-Leibnizator}
		\Pcal_{a+fD}(b+gD,c+hD)=a'(gh+hg)D.
	\end{equation}
	In characteristic zero this is $2a'ghD$.  Taking $a=x$ and $g=h=1$ gives
	\[
	\Pcal_x(D,D)=2D\neq0.
	\]
	Thus this algebra is not Poisson, but it is an $F$-manifold algebra by
	Theorem \ref{thm:LR-to-FMan}.
	
	Set
	\[
	u_r=x^r\in A\quad(r\geq0),
	\qquad
	e_m=x^{m+1}D\in A D\quad(m\geq-1).
	\]
	For $r,s\geq0$ and $m,n\geq-1$,
	\[
	u_r\bullet u_s=u_{r+s},
	\qquad
	u_r\bullet e_m=e_{r+m},
	\qquad
	e_m\bullet e_n=0.
	\]
	In this Witt-type basis, the brackets are
	\begin{align*}
		[e_m,e_n]&=(n-m)e_{m+n}\quad(m+n\geq-1),\\
		[e_m,u_r]&=r u_{m+r}\quad(m+r\geq0),\\
		[u_r,u_s]&=0.
	\end{align*}
	The two boundary cases excluded by these conditions are
	\([e_{-1},e_{-1}]=0\) and
	\([e_{-1},u_0]=0\).
	Moreover,
	\[
	\Pcal_{u_r}(e_m,e_n)
	=r\bigl(e_{m+n+r}+e_{n+m+r}\bigr)
	=2r e_{m+n+r}
	\]
	for \(r\geq1\); for \(r=0\) it is zero.
	
	\subsection{The cotangent \texorpdfstring{$F$}{F}-manifold algebra of a Poisson algebra}
	
	Let $(A,\{-,-\})$ be a Poisson algebra.  Its cotangent Lie--Rinehart algebra
	is $(A,\Omega^1_{A/\kk})$, in which the anchor $\pi^\sharp$ is determined by
	\begin{equation}\label{eq:cot-anchor}
		\pi^\sharp(df)(g)=\{f,g\},
	\end{equation}
	and the Lie bracket $[-,-]_\pi$ is determined by
	\begin{equation}\label{eq:cot-exact}
		[df,dg]_{\pi}=d\{f,g\}.
	\end{equation}
	  Equivalently, for exact generators,
	\begin{equation}\label{eq:cot-bracket-adf}
		[a\,df,b\,dg]_{\pi}
		=ab\,d\{f,g\}+a\{f,b\}\,dg+b\{a,g\}\,df.
	\end{equation}

	The associated $F$-manifold algebra is
	\[
	\Fcot(A):=A\oplus\Omega^1_{A/\kk}.
	\]
	Its product is
	\begin{equation}\label{eq:cot-F-product}
		(a+\alpha)\bullet(b+\eta)=ab+a\eta+b\alpha,
	\end{equation}
	and its bracket is
	\begin{equation}\label{eq:cot-F-bracket}
		[a+\alpha,b+\eta]
		=\pi^\sharp(\alpha)(b)-\pi^\sharp(\eta)(a)+[\alpha,\eta]_{\pi}.
	\end{equation}

	The Leibnizator is
	\begin{equation}\label{eq:cot-Leibnizator}
		\Pcal_{a+\alpha}(b+\eta,c+\gamma)
		=\pi^\sharp(\eta)(a)\gamma+\pi^\sharp(\gamma)(a)\eta.
	\end{equation}
	For exact forms,
	\begin{equation}\label{eq:cot-exact-Leibnizator}
		\Pcal_a(df,dg)=\{f,a\}\,dg+\{g,a\}\,df.
	\end{equation}
	If $\Omega^1_{A/\kk}$ is a finite projective $A$-module of constant rank $r$, then
	Proposition \ref{prop:trace-formula} gives the Poisson bracket in terms of the
	Leibnizator:
	\begin{equation}\label{eq:poisson-from-cotangent-trace}
		\{f,a\}=\frac1{r+1}\operatorname{Tr}\bigl(\eta\mapsto\Pcal_a(df,\eta)\bigr).
	\end{equation}
	Thus, in the finite projective case, the cotangent $F$-manifold algebra
	determines the original Poisson tensor from its Leibnizator.
	
%
%
%

	\subsection{The symplectic affine plane}
	
	Let
	$
	A=\kk[x,y]$ with the Poisson bracket determined by $\{x,y\}=1.$
	Then
	\[
	\rho(dx)=\{x,-\}=\partial_y,
	\qquad
	\rho(dy)=\{y,-\}=-\partial_x,
	\]
	and
	\[
	[dx,dy]_{\pi}=d\{x,y\}=0.
	\]
	The associated $F$-manifold  algebra is
$
	A\oplus A\,dx\oplus A\,dy.
$
	Its bracket satisfies
	\[
	[dx,a]=\partial_y a,
	\qquad
	[dy,a]=-\partial_x a,
	\qquad
	[dx,dy]=0.
	\]
	For two general one-forms
	$
	\alpha=f\,dx+g\,dy,
	~
	\eta=h\,dx+k\,dy,
$
	one has
	\[
	\begin{aligned}
		[\alpha,\eta]_{\pi}
		={}&\bigl(f\partial_yh-h\partial_yf-g\partial_xh+k\partial_xf\bigr)dx +\bigl(f\partial_yk-h\partial_yg-g\partial_xk+k\partial_xg\bigr)dy.
	\end{aligned}
	\]
	The full bracket is therefore
	\[
	\begin{aligned}
		[a+\alpha,b+\eta]
		={}&(f \partial_yb-g \partial_xb-h \partial_ya+k\partial_x a) + [\alpha,\eta]_{\pi}.
	\end{aligned}
	\]
	Finally, the Leibnizator on the coordinate one-forms is
	\begin{equation}\label{eq:symplectic-plane-Leibnizator}
		\Pcal_a(dx,dy)=(\partial_y a)dy-(\partial_x a)dx.
	\end{equation}
	Taking $a=x$ in \eqref{eq:symplectic-plane-Leibnizator} gives
	$\Pcal_x(dx,dy)=-dx\neq0$.  Thus this $F$-manifold algebra is not Poisson.
	
	\vspace{3mm}
	
	\noindent
	\textbf{Acknowledgments.} This research was supported by the Natural Science
	Foundation of Jilin Province (grant no.~20260101013JJ) and the National Natural
	Science Foundation of China (grant nos.~12471060 and W2412041).  The authors
	thank Yunnan Li for helpful discussions.
	
	\smallskip
	\noindent
	\textbf{Declaration of interests.} The authors have no conflicts of interest
	to disclose.
	
	\smallskip
	\noindent
	\textbf{Data availability.} No data were generated or analyzed in this study.

\end{document}